\documentclass[12pt]{article}
\usepackage{subcaption}
\usepackage{graphicx}
\usepackage{amssymb}
\usepackage{amsthm}
\usepackage{amsmath}

\usepackage[utf8]{inputenc}
\usepackage[T1]{fontenc}
\usepackage{lmodern}           
\usepackage{microtype}         
\usepackage{natbib}            
\usepackage[hidelinks]{hyperref}

\title{Mapping Power Relations: A Geometric Framework for Game-Theoretic Analysis}
\author{Daniele De Luca}
\date{September 2025}

\newtheorem{theorem}{Theorem}       
\newtheorem{lemma}{Lemma}           

\newtheorem{corollary}{Corollary}

\usepackage{amsmath}
\begin{document}
\maketitle

\begin{abstract}
This paper develops a geometric framework for analyzing power relations in games. 
We introduce the \emph{preference space}, a canonical domain in which each player’s 
relational stance---altruistic, selfish, antagonistic, or neutral---is represented as a 
vector in a normalized projective space. We show how classical concepts such as bargaining power, dependence, and reciprocity can be recovered within this framework and extended to arbitrary non-cooperative environments. Applications to standard bimatrix games, economic models, and three-player games demonstrate that the preference space reveals structural similarities across different strategic settings, even when their strategy sets differ. In particular, we show how to compute indices that synthesize hierarchical and reciprocal structures at the game level. This unified perspective bridges cooperative and non-cooperative approaches to power, providing a relational measure of influence that is independent of the specific strategic form.
\end{abstract}

\newpage

\section{Introduction}

The study of power relations has deep roots in both economics and political science. Power is often defined as an actor’s capacity to alter another's behavior or welfare. Foundational definitions by \citep{Dahl1957}, \citep{Taylor1982}, \citep{Bardhan1991}, and \citep{Bartlett2006} converge on this core idea of inducing change. Furthermore, as \citep{Bowles1998} and Stoppino in \citep{Matteucci1991} stress, this change is typically directed to serve the power-wielder’s interests. From this perspective, positive (negative) power is the capacity of being favored (harmed) by others. 

Within economics, the Marxist tradition was one of the first to address the theme of power \citep{Ricciardi2019}. For Marx, the capitalist mode of production rests on the appropriation of unpaid labor time, producing a systematic asymmetry of power between classes \citep{Marx2024}. Engels emphasized that this domination is not due to individual greed or wickedness but is embedded in the social organization of production itself \citep{Engels2017}. Lenin extended the analysis to the global sphere, arguing that monopolistic control enables core economies to extract surplus from peripheral nations \citep{Lenin2015}. Later Marxist economists such as \citep{Emmanuel1972} and \citep{Wallerstein2011} developed this into a systematic theory of unequal exchange or dependence, showing how international trade reproduces global hierarchies of power. Across these strands, the Marxist view treats power as a structural inequality masked by the apparent equality of exchange in the free market. This conception is close to the previous one: the worker structurally favors the capitalist, and the peripheral country structurally favors the core.

Classical cooperative game theory introduced power indices such as the Shapley–Shubik index \citep{ShapleyShubik1954} and the Banzhaf index \citep{Banzhaf1965}, which measure a player’s ability to alter collective outcomes in voting situations. Both indices distribute the total “value” or “power” created by a coalition among its members by asking a counterfactual question: what would have happened if this agent had not been present? At first sight, the political–economic tradition and the cooperative game-theoretic tradition appear distinct, if not opposite: for the former, power means to be favored; for the latter, power means the capacity to favor others. Yet both converge in conceiving power as strictly tied to utility and gain. 

With Aristotelian terminology, these two aspects can be described as the actual and the potential of power. In our framework, the preference space has a natural origin at the identity matrix $I$, where subjective and objective utilities coincide. This makes it intrinsically an actual/ex-post construct, representing realized or counterfactually realized payoffs. Potential or ex-ante aspects of power (e.g., capacity to change equilibria before they materialize) can only be recovered by averaging over distributions of preferences, as we do later in the bargaining power section. 

Subsequent research has investigated power in bargaining contexts, often linking it to the distribution of disagreement payoffs, commitment opportunities, or agenda-setting roles \citep{Binmore1987, BaronFerejohn1989, Rubinstein1982}. \citep{Schelling1960} emphasized the role of credible threats as a strategic source of power, while contest theory \citep{Tullock1980, Skaperdas1996} highlighted asymmetries in the conversion of resources into winning probabilities. In parallel, models of informational power have developed within Bayesian persuasion and information design \citep{KamenicaGentzkow2011}. What these diverse strands have in common is the counterfactual reasoning and the attempt to capture power as an influence over outcomes, whether by rules, resources, timing, or threats. Yet none of these approaches provides a general, non-cooperative analysis of power that applies across arbitrary strategic environments.

Describing power is often challenging due to its relational nature: unlike standard economics, you need to consider the choices of a player in connection with the payoff of another player. \citep{GrossiTurrini2012} made a notable attempt to develop a measure of dependence in game theory, raising, for the first time (as long as we know), counterfactual question of the type: what would happen if a player plays assuming the utility of another? Their results build a bridge between non-cooperative and coalitional games.

Along the same line of thought, a recent proposal by \citep{GorlachMotz2024} defines a general measure of bargaining power as counterfactual influence. A player’s power is quantified by the extent to which the outcome distribution of a game changes when their preferences are substituted, normalized relative to the “local dictator” benchmark where all players adopt the same utility. Their axiomatic characterization shows that this measure is unique under natural requirements (null player, local dictator, proportionality, invariance, and compound games). Moreover, when applied to voting environments, it reduces to classical indices such as Shapley–Shubik and Banzhaf, thereby unifying older concepts within a general framework. The key insight is that power can be seen as a structural property of the mapping from preferences to equilibria, not reducible to material payoffs alone or to the credibility of threats.

In the Görlach–Motz model, the set of "relevant" utilities $U$ is chosen on a case-by-case basis: the authors describe it as ‘intentionally vague’ and acknowledge that specifying $U$ is, in practice, a matter of researcher discretion. To counter this arbitrariness, they impose the axiom of invariance to irrelevant extensions (A1), meaning that adding utility functions to $U$ does not change the measure of power. Yet the selection of $U$ remains ex ante discretionary.

While Görlach–Motz avoid the arbitrariness of strategic form by outsourcing to outcome probabilities, they reintroduce arbitrariness in the choice of the utility domain $U$. Our preference space removes this discretion: it provides a canonical, continuous, and game-independent domain where every player’s relational posture is defined up to scale. Unlike discrete specifications of $U$, this domain captures the full geometry of possible relations, including those necessary to sustain mixed-preference equilibria.

In our preference space, each point represents a specific configuration of players' social preferences (altruistic, selfish, antagonistic and so on). Each player may adopt altruistic, selfish, or antagonistic stances toward self and others. This idea has been developed within the tradition of interdependent preference models dating back to \citep{Pareto1913}, \citep{Gaertner1974}, \citep{Pollak1976}, \citep{Bergstrom1987}, and \citep{Bergstrom1999}, and extended to inequity aversion and fairness considerations \citep{FehrSchmidt1999, BoltonOckenfels2000, CharnessRabin2002, AndreoniMiller2002, FismanKarivMarkovits2007}. Empirical studies further stress the heterogeneity and context-dependence of social motives \citep{BlancoEngelmannNormann2011, DimickRuedaStegmueller2018}. Although systems of interdependent utility like the present one have a long history, they were never used to analyze the power relations of games, because they were never used counterfactually.

Here, the analysis of a game with the preference space is articulated in two steps:

\begin{enumerate}
    \item \textit{Projection.} For any given game, we project its payoffs, strategies, or outcomes onto the preference space. 
    This yields a landscape that shows how these elements vary with changes in the relational structure between players.
    
    \item \textit{Reduction and Analysis.} The resulting landscape can then be studied by focusing on key features:
    \begin{itemize}
        \item \textit{Local Analysis:} examining specific, strategically significant profiles 
        (e.g., cardinal and coalitional points, potential power) and their associated outcomes.
        \item \textit{Global Analysis:} identifying a representative point in the preference space, 
        the center of mass, which summarizes the overall network structure.
        \item \textit{Structural Indices:} defining indices (Reciprocity, Hierarchy, Diagonalization) 
        that provide a quantitative description of the power relations at any given point in this space.
    \end{itemize}
\end{enumerate}

\section{Materials and Methods}

\subsection{Overview}
In what follows, we present several preference-space-based analyses of games. All computations and figures are produced on Python. In the first section on strategic bimatrix games, we employ the \textit{support enumeration} 
algorithm from the \textit{Nashpy} Python library to compute equilibria, and we average 
only Pareto-efficient equilibria to compute players’ payoffs at sampled points in 
the preference space.  We use \textit{Matplotlib} for plots.
In the analysis of the Cournot model, we use a heuristic approximation method that 
starts from an arbitrary point (0) and proceeds by iteratively finding the reciprocal 
best responses of each player for a fixed number of steps, or terminates if both players 
no longer wish to change their strategies. 
For the case of three players with two strategies each, we adopt a Monte Carlo sampling 
scheme of the preference space. For each run, equilibria are computed at the sampled 
points and the corresponding centers of mass are evaluated; repeating the procedure with 
different random seeds allows us to estimate the statistical error.

\subsection{The Outcome Mapping Function $\mu$}
\label{subsec:mu_function}

The function $\mu$ is a fundamental component of our framework, serving as the \textit{outcome mapping function}. It translates the vector of subjective payoff functions, $V\mathbf{u}$, into a probability distribution over the set of outcomes, $O$. Formally, $\mu(V\mathbf{u}) \in \Delta(O)$, where $\Delta(O)$ is the simplex of probability distributions over $O$.

The interpretation of $\mu$ is that it represents the \textit{strategic response} of the players to their subjective preferences. In other words, $\mu$ encapsulates a solution concept that predicts how the game will be played given the profile of interdependent preferences $V$. To ensure generality and avoid commitment to a single, potentially restrictive solution concept, we define $\mu$ axiomatically by its essential properties rather than by a specific algorithmic procedure. This allows the framework to be applied across a wide range of strategic environments.

The following axioms are imposed on $\mu$:

\begin{enumerate}
    \item \textit{Well-Definedness:} $\mu(V\mathbf{u})$ is defined for every $V \in \mathbb{V}$, and the expected payoff $\mathbb{E}_i(V)$ is defined for each player $i$.
    
    \item \textit{Nash Consistency:} The support of $\mu(V\mathbf{u})$ must be a subset of the Nash equilibria of the game defined by the subjective payoff functions $V\mathbf{u}$. That is, $\text{Supp}(\mu(V\mathbf{u})) \subseteq \text{NE}(V\mathbf{u})$. This ensures that the predicted outcomes are strategically stable given the players' preferences.
    
    \item \textit{Pareto Efficiency:} We require that $\mu(V\mathbf{u})$ is supported only on the Pareto-efficient Nash equilibria relative to the set $\text{NE}(V\mathbf{u})$. This is a natural refinement because we are dealing with the capacity of players to play together, for example in a coalition or in a dictator position. This axiom assure coordination among players.
    
    \item \textit{Invariance to Positive Affine Transformations:} The distribution $\mu(V\mathbf{u})$ is invariant under positive affine transformations of the subjective utility functions. This aligns with the standard assumption in game theory that utility is ordinal, not cardinal.
\end{enumerate}

\section{Results}

\subsection{Projection onto the Preference Space}

\subsubsection{Requirements}

To analyze power and influence in games, we require a formal language capable of representing the full spectrum of possible attitudes—altruism, selfishness, antagonism, and indifference—that players may hold towards one another. These attitudes are fundamentally relational and directional: how much does player $i$ favor or harm player $j$? We thus postulate that any representation of relational power must satisfy the following requirements. 

\begin{enumerate}
    \item The framework must represent the attitude of any player $i$ towards any player $j$ by a real number $v_{ij}\in\mathbb{R}$. The sign of $v_{ij}$ indicates the direction of the attitude (positive for favor, negative for harm), and its magnitude $|v_{ij}|$ indicates the intensity. This is the most basic requirement: it encodes a weighted, directed graph of affinities among players, which is the bedrock of sociological network analysis and relational power. The attitudes of $n$ players are represented by a matrix $V\in\mathbb{R}^{n\times n}$. By convention, each vector $\mathbf{v}_i$, which represents the importance assigned by player \( i \) to the payoffs of all players, including themselves, is taken to be a row vector. 
    \item Only the \emph{relative} intensity of a player’s attitudes towards others is meaningful. A statement such as “I favor you twice as much as I favor her” is meaningful; stating “my favor has an absolute intensity of 2” is not. The framework must be invariant under positive scalar transformations of any player’s vector of attitudes. Formally, for any player $i$ and any $\lambda>0$, the attitude vectors $\mathbf{v}_i$ and $\lambda \mathbf{v}_i$ represent the same relational posture.
    \begin{equation}
    \mathbf{v}_i \sim \lambda \mathbf{v}_i, \quad \forall \lambda > 0.
    \end{equation}
    That is, each vector $\mathbf{v}_i$ is defined up to a positive scalar multiple. If $n$ is the number of players, the preference of each player $i$ is an element of the projective space $\mathbb{R}^n/\!\sim$, which is isomorphic to the $(n-1)$-sphere $S^{n-1}$ upon normalization ($\|\mathbf{v}_i\|=1$). Thus the full preference space is $(S^{n-1})^n$. 
\end{enumerate}

We refer to this space as the \textit{preference space} (or $\mathbb{V}$) and to any point within this space as a \textit{preference position} $V$. Each $V$ can be viewed also as a weighted directed graph, so we will use alternatively terminology from graph theory to address players and power relations, i.e., "nodes" and "edges".

Directionality of preferences is both sufficient and necessary for a canonical representation of power. It is sufficient, because any relational attitude between players—favoritism, antagonism, indifference—can be expressed as a directed weight from one to the other, and deviations from efficient equilibria can be understood as quantifiable shifts in such directional preferences. It is also necessary, because restricting attention to discrete stances leaves out important cases: there exist outcomes that are equilibria only under mixed preference configurations (see for example the discussion on Figure \ref{fig:BS_outcome}).

\subsubsection{General Structure of a Game}
A game \( G \) is defined by the tuple:
\begin{equation}
\Gamma = (N, O, \mathbf{\mathbf{u}}, \mu),
\end{equation}
where:
\begin{itemize}
    \item \( N \) is the set of players, \( |N| \geq 2\),
    \item \( O \) is the set of outcomes, \(|O| \geq 2\),
    \item $\mathbf{\mathbf{u}}$ is the vector of \textit{objective payoff functions}, one per player, \( \mathbf{u}_i(o): O \to \mathbb{R}^+ \), that define the payoff received by player \( i \in N \) for the outcome \(o \in O\), and that we always assume as positive valued,
    \item $\mu$ is a function that maps any vector of payoff functions to a probability distribution into the outcome set. It represents a solution concept satisfying specific axioms, as detailed in Section~\ref{subsec:mu_function}.
\end{itemize}

The \textit{subjective payoff function} $(V \mathbf{u})_i : O \rightarrow \mathbb{R}$ is derived from some preference position $V \in \mathbb{V}$ and the vector of objective payoff functions $\mathbf{u}$ by matrix multiplication, where $(\cdot)_i$ denotes the $i$-th row. In explicit notation:
\begin{equation}
(V \mathbf{u})_i=\sum_j v_{ij} \mathbf{u}_j(o)
\end{equation}
This function captures how player $i$ evaluates outcomes based on their own preferences and attitudes toward others. These attitudes are not fixed: they can vary within $\mathbb{V}$. If $V=I$ \textit{objective} and \textit{subjective} payoff coincide.

The key assumption is that the \textit{distribution $\mu$ is induced by the subjective payoff function $V\mathbf{u}$}. 
In other words, outcomes that yield higher subjective value for the players are assigned higher probability, 
reflecting the idea that strategic behavior is driven by preference-weighted expectations.

By plotting $\mu(V\mathbf{u})$ into the preference space $\mathbb{V}$ we can classify the outcomes 
according to their position in this space. 
Since each outcome corresponds to a specific strategy profile, for any non-null player we can therefore also map strategies. 

Moreover, we can project the objective payoff function $\mathbf{u}_i$ itself into the preference space. 
The expected (objective) payoff of player $i$ under $\mu$ is then given by
\begin{equation}
\mathbb{E}_i(V) \;=\; \int_O \mathbf{u}_i(o)\, d\mu_o(V\mathbf{u}),
\end{equation}
which captures how the relational structure encoded in $V$ shapes the actual material payoffs of the game.

We follow \citep{GorlachMotz2024} assuming the distribution of outcome probabilities as given, though in Theorem~\ref{th:expectedmax} we attribute some general properties to the function $\mu$. Outsourcing the strategic component, we transforms every non-cooperative game into an extended coalition game. In our framework, coalitions are encoded as mutual networks of preferences, and the characteristic function arises endogenously from the geometry of the preference space. Classical coalition games are thus included in $\mathbb{V}$ as a reduction, corresponding to profiles where players adopt purely mutual preferences (coalitional vectors). More generally, however, $\mathbb{V}$ provides a continuum of relational stances—altruistic, antagonistic, neutral—so that every non-cooperative game can be studied as a coalition game with extended domains of cooperation and conflict.

For $C\subseteq N$, let $\mathbf{u}^C:=\sum_{j\in C} \mathbf{e}_j\in\mathbb{R}^n$, where $\{\mathbf{e}_j\}$ are canonical basis vectors.
The \emph{preference position} for $C$ is $V^{C}$ defined by
$\mathbf{v}_i = \mathbf{u}^C/\|\mathbf{u}^C\|$ for all $i\in C$, and $\mathbf{v}_k=e_k$ for $k\notin C$.

For any coalition $C\subseteq N$, we define the characteristic function $v(C)$
\begin{equation}
v(C) \;:=\; \sum_{i\in C} \mathbb{E}_i\!\big(V^{C}\big)
\end{equation}

For any non-cooperative environment $\Gamma=(N,O,\mathbf{u}, \mu)$, the pair
\begin{equation}\label{coalitional}
\big(N,\, v(C)\big)
\end{equation}
is a cooperative coalition game canonically induced by the preference space. In particular, classical cooperative solution concepts (core, Shapley value) apply.

\subsection{Analysis of the Preference Space}

\subsubsection{Cardinal and Extreme Points}

The set of \textit{cardinal points} $K \subset \mathbb{V}$ is the set of points where each player's preference vector is a canonical basis vector (or its negative):
\begin{equation}
K=\{ V \in \mathbb{V} : \forall j \in  N, \exists k \in N, (\mathbf{v}_j=\mathbf{e}_k) \, \lor \, (\mathbf{v}_j=-\mathbf{e}_k) \}
\end{equation}
This set contains all the points in $\mathbb{V}$ where the preference positions are discrete, i.e., each player can favor or harm only one of the other players or themselves. Note that $I \in K \subset \mathbb{V}$ is the identity matrix, that is the only preference position where objective and subjective payoff coincide. $K_i$ is the set of cardinal preferences of player $i$.

Let $\mathbf{1}$ denote a $n$-dimensional column vector of all ones, $n=|N|$. Then $\mathbf{1} \mathbf{h}^\top \in \mathbb{V}$ is a matrix with all the rows equal to some vector $\mathbf{h}$.  
$\mathbf{1} \mathbf{e}_j^\top$ ($-\mathbf{1} \mathbf{e}_j^\top$)is the expected maximum (minimum) cardinal point of $\mathbb{V}$ for $j$, i.e., where every player’s subjective preference is aligned to maximize (minimize) player $i$’s objective payoff. 
The following proposition states that $\mathbf{1} \mathbf{e}_j^\top$ is always a global maximum if the $\mu$ function considers only Pareto-efficient equilibria, as required in Section~\ref{subsec:mu_function}.

\begin{theorem}\label{th:expectedmax}[Expected Maximum Point]
Let $\mathrm{PE}(V)\subseteq \mathrm{NE}(V)$ be the subset of equilibria that are Pareto–optimal with respect to other equilibria in $\mathrm{NE}(V)$. Define the expected payoff of player $j$ at $V$ as
\begin{equation}
\mathbb{E}_j(V) \;=\; \int_O \mathbf{u}_j(o)\, d\mu_o(V\mathbf{u}),
\end{equation}
where $\mu$ is a probability measure supported on $\mathrm{PE}(V)$. Then $\mathbb{E}_j(V)$ attains its global maximum at $V^\star=\mathbf{1}\mathbf{e}_j^\top$.
\end{theorem}

\begin{proof}
Suppose, to the contrary, that $V^\star=\mathbf{1}\mathbf{e}_j^\top$ is not a global maximizer. Then there exists $V'$ such that $\mathbb{E}_j(V')>\mathbb{E}_j(V^\star)$. Let $M=\max_o u_j(o)$.

At $V^\star$, all players share the same subjective utility, namely $u_j(o)$. The transformed game is thus an identical–interest game with common payoff $u_j$. Any outcome $o^\dagger$ satisfying $u_j(o^\dagger)=M$ must be a Nash equilibrium: if some player had a profitable deviation, the deviation would yield a payoff strictly greater than $M$, contradicting maximality. Therefore, all global maximizers of $u_j$ are Nash equilibria of the game at $V^\star$. 

It remains to prove that all equilibria at $V^\star$ are global maximizers. Since all players’ utilities coincide at $V^\star$, Pareto dominance is simply comparison of $u_j(o)$. Thus, among the Nash equilibria, the Pareto–optimal ones are exactly those outcomes $o^\dagger$ with $u_j(o^\dagger)=M$. Every such equilibrium yields the same payoff $M$ to player $j$. As $\mu$ is supported on this set, it follows that $\mathbb{E}_j(V^\star)=M$.
\end{proof}

\begin{corollary}\label{co:expectedmin}[Expected minimum point]
Within the same condition as Theorem~\ref{th:expectedmax}, $\mathbb{E}_j(V)$ attains its global minimum at $V^\star=-\mathbf{1}\mathbf{e}_j^\top$.
\end{corollary}

\begin{proof}
The argument is symmetric to the proof of Theorem~\ref{th:expectedmax}, 
applied to the identical–interest game with common payoff $-u_j$.
\end{proof}

\subsubsection{Bargaining Power}

We can measure the overall capacity of player $i$ to favor or harm player $j$, i.e., the relational potential power of $i$ on $j$, with:
\begin{equation}
\tilde\rho_{ij}
\;=\;
\frac{\displaystyle \int_{V_{-i}}
\Big(\mathbb{E}_j(V_{-i \gets \mathbf{e}_j})-\mathbb{E}_j(V_{-i \gets -\mathbf{e}_j})\Big)}
{\;\mathbb{E}_j(\mathbf{1}\mathbf{e}_j^\top)-\mathbb{E}_j(-\mathbf{1}\mathbf{e}_j^\top)\;}.
\end{equation}
Here the integral is taken with respect to the uniform distribution over all configurations of the other players’ preferences $(V_{-i})$. The terms $V_{-i \gets \mathbf{e}_j}$ and $V_{-i \gets -\mathbf{e}_j}$ means that the $i$-th row is respectively $\mathbf{e}_j$ and $-\mathbf{e}_j$, while other rows can vary over the entire space.  We consider also a local version of the previous formula and we show that, in classical voting games, it coincides with \textit{pivotality}:
\begin{equation}\label{rho}
\rho_{ij}= \frac 
{\mathbb{E}_j(I_{\mathbf{e}_i \gets \mathbf{e}_j}) - \mathbb{E}_j(I_{\mathbf{e}_i \gets \mathbf{-e}_j})}
{\mathbb{E}_j(\mathbf{1} \mathbf{e}_j^\top) - \mathbb{E}_j(-\mathbf{1} \mathbf{e}_j^\top)},
\end{equation}
where $I_{\mathbf{e}_i \gets \mathbf{e}_k}$ is the identity matrix (representing the baseline of selfish preferences) with the $i$-th row replaced by the vector $e_k$, meaning player $i$'s preferences are altered to purely and exclusively favor player $j$.

The numerator of this measure quantifies the total local causal effect player $i$ could have on player $j$'s payoff. It captures the full range of change in $j$'s expected objective payoff as $i$'s stance flips from one extreme of pure favoritism ($e_j$) to the other extreme of pure harm ($-e_j$), holding all other players' preferences at their baseline.

The denominator establishes the \textit{total possible range} for player $j$'s payoff within the game. It represents the difference between $j$'s payoff in the expected maximum and minimum cardinal point for $j$. This serves as a natural benchmark, normalizing $i$'s influence against the maximum possible influence that could ever be exerted on $j$'s welfare by the entire group.

The formula \ref{rho} is valid if and only if
\begin{equation}
\mathbb{E}_j(\mathbf{1} \mathbf{e}_j^\top) - 
\mathbb{E}_j(-\mathbf{1} \mathbf{e}_j^\top) \neq 0,
\end{equation}
and thanks to Theorem~\ref{th:expectedmax} and Corollary~\ref{co:expectedmin} we know the conditions where this requirement always holds, unless the payoff of player $j$ is degenerate, i.e. it is constant.

It is worth noting that \ref{rho} can be diagonalized as:
\begin{equation}\label{diagrho}
\rho_{ii}= \frac 
{\mathbb{E}_i(I) - \mathbb{E}_i(I_{\mathbf{e}_i \gets \mathbf{-e}_i})}
{\mathbb{E}_i(\mathbf{1} \mathbf{e}_i^\top) - \mathbb{E}_i(-\mathbf{1} \mathbf{e}_i^\top)}.
\end{equation}

With $\rho_{ii}$ we are measuring the bargaining power of $i$ on themselves, relating the magnitude of the difference in $i$'s payoff if \emph{only} $i$ passes from selfishness to self-harming, to the same magnitude if all the players make the same passage.

Thus, $\rho_{ij}$ can be interpreted as the proportion of the total possible variation in player $j$'s payoff that player $i$ can single-handedly control through changes in their relational stance.

This interpretation is clarified by considering its extreme values:
\begin{itemize}
    \item $\rho_{ij} = 1$: This occurs if player $i$ has complete and dictatorial control over player $j$'s welfare. In this case, $i$'s preference shift from pure harm to pure favoritism moves $j$'s payoff from its absolute minimum to its absolute maximum. The preferences of all other players are irrelevant for $j$'s outcome; $i$'s will alone determines it. $i$ is a \emph{local dictator} with respect to $j$.

    \item $\rho_{ij} = 0$: This occurs if player $i$ has no influence whatsoever on player $j$'s payoff. Changing $i$'s stance from pure harm to pure favoritism has no effect on the expected value of $j$'s objective outcome. $j$'s payoff is entirely determined by the other players and the structure of the game, making $i$ a \emph{null player} with respect to $j$.

    \item $\rho_{ij} < 0$: This is a paradoxical situation where the aim of the action of a player is in contradiction with the effects it produces. We will treat this possibility in the Application section and call it the \textit{Self-Harm Paradox}.
\end{itemize}

The diagonalized form of the relational potential power measure, \ref{diagrho}, provides a direct bridge to classical cooperative power indices when applied to voting games. In this context, where payoffs are binary and the denominator simplifies to unity, $\rho_{ii}$ reduces to an indicator function that takes a value of 1 if player $i$ is pivotal—that is, if a change in their vote alters the outcome—and 0 otherwise.

\begin{lemma}\label{lemma:pivotality}[Pivotality in Voting Games]
Consider a binary voting game with outcome set $O=\{0,1\}$, quota $q>0$, and weights $w\in\mathbb{R}^N_{+}$.  
The \emph{objective payoff} of player $i$ is
\begin{equation}
\mathbf{u}_i(o) \;=\;
\begin{cases}
1 & \text{if the collective outcome $o=1$ coincides with $i$'s vote,}\\[6pt]
0 & \text{otherwise}.
\end{cases}
\end{equation}

The diagonal component of the relational potential power measure is equation \ref{diagrho}. In a voting game, the denominator equals $1$, while the numerator equals $1$ if and only if flipping player $i$’s vote changes the collective outcome. Therefore,
\begin{equation}
\rho_{ii} =
\begin{cases}
1 & \text{if $i$ is pivotal},\\[6pt]
0 & \text{otherwise}.
\end{cases}
\end{equation}
\end{lemma}

This mirrors the core logic of marginal contribution in voting power theory. To transition from this ex-post, preference-specific measure of influence to an ex-ante measure of a player's inherent power within the voting rule itself, one must average $\rho_{ii}$ over a distribution of all possible preference profiles. As established by \citep{GorlachMotz2024}, employing a uniform distribution (impartial culture) over these profiles yields an expectation $\mathbb{E}[\rho_{ii}]$ that is proportional to the Penrose-Banzhaf index. Similarly, averaging over a distribution where all coalition sizes are equally likely recovers a value proportional to the Shapley-Shubik index.

\begin{theorem}[Equivalence with Classical Voting Power Indices]\label{th:voting}
In voting games, by Lemma~\ref{lemma:pivotality} and \citep{GorlachMotz2024}'s proposition~2, averaging $\rho_{ii}$
over preference distributions recovers the standard cooperative power indices:
\begin{itemize}
\item under the impartial culture, $\mathbb{E}[\rho_{ii}]$ is proportional to the Penrose--Banzhaf index;
\item under a distribution where all coalition sizes are equally likely, $\mathbb{E}[\rho_{ii}]$ is proportional to the Shapley--Shubik index.
\end{itemize}
\end{theorem}

The diagonal bargaining power measure $\rho_{ii}$ presented here offers a conceptual and computational simplification over the formulation in \citep{GorlachMotz2024}. While their axiomatic approach requires integration over a potentially vast space of utility functions, our measure is determined by a limited set of strategically meaningful points within a well-defined preference space. This construction not only makes the measure more tractable but also embeds it within a richer geometric framework that explicitly models the network of affinities between players.

Although, a close translation of those classical power indices in the preference space is possible only through the coalitional game conversion (\ref{coalitional}).

\subsubsection{Center of Mass (CoM)}

The Center of Mass (CoM) of any function defined on the preference space $\mathbb{V}$ 
is a key measure that summarizes the positioning of that function within $\mathbb{V}$. 
In general, it is defined as the weighted expectation of preference positions, 
where the weight is given by the function itself.

\begin{theorem}[Well-definedness of the Center of Mass]
Let $\mathbb V\subset \mathbb R^{n\times n}$ be the preference space.  
Suppose $f:\mathbb V\to [0,M]$ with $M>0$ is measurable and not identically zero. Then
\begin{equation}\label{CoM}
V_f^M \;=\; \frac{\displaystyle \int_{\mathbb V} V f(V)\,dV}{\displaystyle \int_{\mathbb V} f(V)\,dV}
\end{equation}
is well defined as an element of $\mathbb R^{n\times n}$.
\end{theorem}

\begin{proof}
Since $f$ is bounded by $M$ and $\mathbb V$ is compact, both integrals are finite.  
Moreover, because $f$ is not identically zero and always positive, the denominator is strictly positive.  

Each $V\in\mathbb V$ has rows of unit length, hence $\|V\|_F=\sqrt{n}$ for all $V$, so the numerator’s integrand $V f(V)$ is uniformly bounded. Thus the numerator is finite.  

Therefore the quotient is well defined in $\mathbb R^{n\times n}$.
\end{proof}

The CoM is no longer a unit-vector matrix. The norm of each row vector of a CoM is an important measure of its strength. When this norm approaches the estimated computational error, we say that the row is \textit{degenerate}, meaning that it is equally balanced in all directions.

The CoM of the \emph{expected payoff} of player $i$ identifies the average preference configuration weighted by the objective payoff 
that player $i$ receives across the preference space, reflecting overall strategic tendencies. Specifically, we are interested in analyzing the shift of a player's $V_{\mathbb{E}_i}^M$ from $\mathbf{1} \mathbf{e}_j^\top$ (the expected maximum point for $i$). 

The CoM can also be computed for a \emph{single outcome} $o$. 
In this case, the $\mu_o(V)$ function is projected onto the preference space ($V_{\mu_o}^M$) and the CoM identifies the average preference configuration that sustains it. Two outcomes of the same game may therefore have 
very different CoMs: cooperative outcomes tend to cluster in regions of high reciprocity, 
while exploitative outcomes typically fall into hierarchical or antagonistic regions. 

Similarly, one can define the CoM of a \emph{strategy}. Here the CoM summarizes the 
distribution of preference profiles under which a given strategy is played. 
The self-component of this CoM describes the type of relational stance (altruistic, selfish, antagonistic) 
that most characterizes the strategy for the player who adopts it. 
The other-component, instead, reveals the context in which the strategy is typically sustained: 
for example, whether a player cooperates only when the opponent is cooperative, or also in the face of 
antagonistic preferences. In this sense, the CoM of a strategy provides a “relational fingerprint” 
of the player’s behavior, showing both its internal motivation and the external conditions that support it.

\subsubsection{General Indices in the Preference Space}

We need general indices for classifying any $V \in \mathbb{V}$: diagonalization $D$, hierarchy $H$, and reciprocity $R$. 
These indices must satisfy the following requirements:

\begin{enumerate}
  \item The $D$ index must reflect the auto-orientation of edges in the graph (self-loops), 
  and attain its maximum if and only if $V$ is diagonal.
  
  \item The $H$ index must reflect the asymmetry in power between pairs of nodes, 
  and attain its maximum if and only if $V$ is antisymmetric.
  
  \item The $R$ index must reflect the symmetry and mutualism between pairs of nodes, 
  and attain its maximum if and only if $V$ is symmetric with zero diagonal.
  
  \item If the weights of the edges between two nodes have opposite signs, i.e.\ if one favors the other while the other harms the first, 
  then their contribution to $R$ must be 0. 
  If instead they have the same sign, their contribution to $R$ must be proportional 
  to the part of the weights they have in common.
  
  \item The indices must be continuous function of $V$ and mutually exclusive, 
  so that $D + H + R = k$, for some constant $k$ independent of $V$.
\end{enumerate}

The following are the simplest indices that satisfy these properties.  
Remember that any $V \in \mathbb{V}$ is an $n \times n$ preference matrix with unit row vectors 
($\sum_k v_{ik}^2 = 1$). 
Define the square matrix $W$ by $w_{ij} = v_{ij}^2$ and the signed square matrix $\widetilde{W}$ by 
\begin{equation}
\tilde{w}_{ij} = \operatorname{sign}(v_{ij}) \, w_{ij}.
\end{equation}

The indices of Diagonalization ($D$), Hierarchy ($H$), and Reciprocity ($R$) are defined as
\begin{align*}
D &= \frac{1}{n} \sum_{i=1}^{n} w_{ii}, \\
H &= \frac{1}{n} \sum_{i < j} \left| \tilde{w}_{ij} - \tilde{w}_{ji} \right|, \\
R &= \frac{1}{n} \sum_{i < j} \left( w_{ij} + w_{ji} - \left| \tilde{w}_{ij} - \tilde{w}_{ji} \right| \right).
\end{align*}

The indices are non-negative and form a complete partition of the total relational structure:
\begin{equation}
R + H + D = 1, 
\qquad \text{with} \quad 0 \leq R, H, D \leq 1.
\end{equation}
$D$ captures the degree to which players focus on themselves (self-referential strategies), rather than being influenced by or influencing others.
$H$ measures asymmetry in strategic relationships, revealing dominance or directional imbalance in power.
$R$ measures the overall symmetric interaction between players. High values indicate mutual reinforcement or alignment in preferences.

Note that for two corresponding off-diagonal elements of $\tilde{W}$, 
\[
w_{ij} + w_{ji} - \left|\tilde{w}_{ij} - \tilde{w}_{ji}\right|
\]
is equal to twice the minimum of the absolute values of the two, if they share the same sign, 
and equal to $0$ if they have opposite signs. 
This satisfies the fourth property: if the preferences that two players have toward each other 
are opposite, there is no reciprocity; if they share the same sign, 
the reciprocal part is the minimum (in absolute value) between $w_{ij}$ and $w_{ji}$ 
for both players.

It is also possible to distinguish positive from negative reciprocity, i.e.\ 
cooperation versus antagonism, by splitting the index into two sub-indices:
\[
R = R^+ + R^-,
\]
where $R^+$ ($R^-$) contains only the positive (negative) off-diagonal coupled entries.

\subsubsection{Global Center of Mass of the Payoffs (GCoM)}

A direct application of the indices $D,H,R$ to the CoM of a player’s
expected payoff is not informative, since by Theorem~\ref{th:expectedmax} the expected maximum
point $\mathbf{1}\mathbf{e}_j^\top$ is always a global maximizer. In particular, we have $H(\mathbf{1}\mathbf{e}_j^\top)=1$, but the fact that the CoM lies close to the expected maximum point is not in itself evidence of hierarchy in the game. It is trivial to observe that each player desires
to be a dictator, and this trivial fact should not be confused with a meaningful structural property of the interaction.

The CoM of a player’s payoff represents the direction in the preference space 
toward which that player would like to move every other player (and themselves) 
in order to maximize their own payoff in the local neighborhood of the space. The only row of $V_{\mathbb{E}_i}^M$ aligned with the center $I$ is the $i$-th row, that records precisely this orientation for the player
in question. If we interpret payoff, in line with standard economic theory, as
an ex-post quantity---representing either the player’s actual choices or their
possible choices under counterfactual hypotheses---then the self-row of the
CoM provides a clear ex-post representation of the attitude exhibited by that
player toward others.

Therefore, if we want to apply the indices to the CoM of payoff functions,
we must construct a global CoM that gathers only the self-rows. In this way
we obtain a synthetic measure of the attitudes of players toward one another,
filtering out the trivial component that each player maximizes their own payoff.

Formally, let $\mathbf{e}_i \mathbf{e}_i^\top$ denote the diagonal unit matrix
with only the $i$-th row nonzero.
Then we define the \emph{global center of mass of the payoffs} (GCoM) as
\begin{equation}
V_{\mathbb{E}}^{M}
\;:=\;
\sum_{i=1}^{n} \mathbf{e}_i \mathbf{e}_i^\top V_{\mathbb{E}_i}^M.
\end{equation}
When applying the indices $D$,$H$,$R$, we interpret each row of $V_{\mathbb{E}}^{M}$ up to scalar normalization, in line with the projective nature of the preference space.

\subsubsection{Permutation Points and Reciprocity}
\label{sec:permutation}

In the framework of Grossi--Turrini, any player permutation $\xi$ can be viewed
as a transformation where each player $i$ evaluates outcomes via $u_{\xi(i)}$.
Within our preference space $\mathbb{V}$, these transformations correspond to
discrete \emph{permutation points} $V^\xi$, whose rows are canonical basis
vectors permuted according to $\xi$. Cyclic permutations encode \emph{reciprocal}
structures: for a $k$-cycle $\xi=(i_1\,i_2\,\dots\,i_k)$, the rotations
$\xi^t$ $(t=0,\dots,k-1)$ implement the idea that each $i_\ell$ plays ``for''
its successor $\xi(i_\ell)$. As shown by Grossi--Turrini, reciprocity is characterized by equilibrium in \emph{all}
rotations of the cycle. 

We will prove that, when equilibria occur at multiple
permutation points that “split” a pair $(i,j)$ across two directions
($i\!\to\!j$ and $j\!\to\!i$), the GCoM accumulates symmetric
off-diagonal mass and thus yields a strictly positive contribution to the
reciprocity index $R$ of the game.

\begin{lemma}[Permutation points in $\mathbb{V}$]
Let $\xi:N\to N$ be a permutation of players. The matrix
\[
V^\xi \;=\;
\begin{bmatrix}
\mathbf{e}_{\xi(1)} \\
\vdots \\
\mathbf{e}_{\xi(n)}
\end{bmatrix}
\in \mathbb{V}
\]
is a cardinal point of the preference space. For each player $i$,
\(
(V^\xi \mathbf{u})_i = \mathbf{u}_{\xi(i)}.
\)
\end{lemma}

\begin{proof}
The $i$-th row of $V^\xi$ is $\mathbf{e}_{\xi(i)}$, hence $(V^\xi \mathbf{u})_i=\sum_k (\mathbf{e}_{\xi(i)})_k \mathbf{u}_k=\mathbf{u}_{\xi(i)}$.
\end{proof}

\begin{lemma}[Convex reciprocity from split equilibria for a pair]
Fix distinct players $i\neq j$. Assume there exist two preference points
$V^{i\to j}, V^{j\to i}\in\mathbb{V}$ such that
\[
\text{row}_i(V_{i\to j})=\mathbf{e}_j,\qquad \text{row}_j(V_{j\to i})=\mathbf{e}_i,
\]
and an outcome $\bar o$ which is an equilibrium
at \emph{both} points (the rows of the other players may be fixed—possibly
differently—at the two points). Let $\mu$ be any probability measure used to
compute a global center of mass $\tilde{W}_{\mathbb{E}}^M$ that assigns
strictly positive weight to neighborhoods of both $V_{i\to j}$ and $V_{j\to i}$.
Then the pairwise contribution of $(i,j)$ to the reciprocity index of the GCoM
matrix $V^M$ is strictly positive:
\[
R_{(i,j)}(V^M)\;=\;\frac{1}{n}\Big(w^M_{ij}+w^M_{ji}-\big|\tilde w^M_{ij}-\tilde w^M_{ji}\big|\Big)\;>\;0.
\]
\end{lemma}

\begin{proof}
Start from the definition of CoM, \ref{CoM}.
Because $f$ puts positive
mass near $V_{i\to j}$ (where $v_{ij}=1$) and near $V_{j\to i}$ (where $v_{ji}=1$),
we get $v^M_{ij}>0$ and $v^M_{ji}>0$. Hence $w^M_{ij}=(v^M_{ij})^2>0$ and
$w^M_{ji}=(v^M_{ji})^2>0$, with equal (positive) sign; therefore
\(
w^M_{ij}+w^M_{ji}-|\tilde w^M_{ij}-\tilde w^M_{ji}|=2\min\{w^M_{ij},w^M_{ji}\}>0,
\)
yielding $R_{(i,j)}(V^M)>0$.
\end{proof}

\begin{lemma}[Cycle-based reciprocity and permutation points]
Let $\xi=(i_1\,i_2\,\dots\,i_k)$ be a dependence cycle in the sense of Grossi--Turrini. 
An outcome $o$ is reciprocal
for the cycle if and only if, for every $t=0,\dots,k-1$, $o$ is an
equilibrium at the permutation point $V^{\xi^t}$ (players in the cycle
evaluate outcomes via $u_{\xi^t(i)}$, while the other players keep their rows
fixed as prescribed by the dependence construction). In particular, for any
adjacent pair $(i_\ell,i_{\ell+1})$ in the cycle, $o$ is an equilibrium
when $i_\ell$ plays for $i_{\ell+1}$ given the positions of the other players in
that rotation.
\end{lemma}

\begin{proof}
($\Rightarrow$) If $o$ is reciprocal for the cycle, Grossi--Turrini’s
Theorem~1 states that $o$ is an equilibrium in every game obtained by
permuting utilities along the cycle; evaluating at $V^{\xi^t}$ reproduces exactly
those permuted utilities, hence $o$ is an equilibrium at each $V^{\xi^t}$.

($\Leftarrow$) Conversely, if $o$ is an equilibrium at all $V^{\xi^t}$,
then it is an equilibrium in each corresponding cycle-permuted game; by
Grossi--Turrini’s equivalence, $o$ is reciprocal for the cycle.
\end{proof}

The three lemmas above connect Grossi--Turrini’s notion of reciprocity with the
geometry of the preference space. In particular, they show that whenever an outcome
is reciprocal in the sense of Grossi--Turrini, its equilibria appear in permutation
points that jointly generate a strictly positive contribution to the reciprocity index.
This establishes that the index $R$ applied to the global center of mass is a
substantive measure of mutual dependence in equilibrium analysis.

\begin{theorem}[Reciprocal outcomes and reciprocity of GCoM]
Every reciprocal outcome in the sense of Grossi--Turrini contributes
strictly positively to the reciprocity index $R$ of the global center of mass of
the players’ payoffs.
\end{theorem}

\begin{proof}
By Lemma 3, reciprocal outcomes are equilibria at all permutation points
along a cycle. By Lemma 2, equilibria appearing at both $V_{i\to j}$ and
$V_{j\to i}$ ensure $v^M_{ij},v^M_{ji}>0$, hence a strictly positive term
in $R$. Therefore any reciprocal outcome yields a positive reciprocity
contribution in the GCoM.
\end{proof}

\subsubsection{Two players restriction}

In the case of two players, the preference space reduces to the product of two circles \(S^1 \times S^1\) or a torus (see Figure~\ref{fig:2pRPS}). Let the two players be denoted by $a$ and $b$.

\begin{figure}
\subfloat[\centering]{\includegraphics[width=7.0cm]{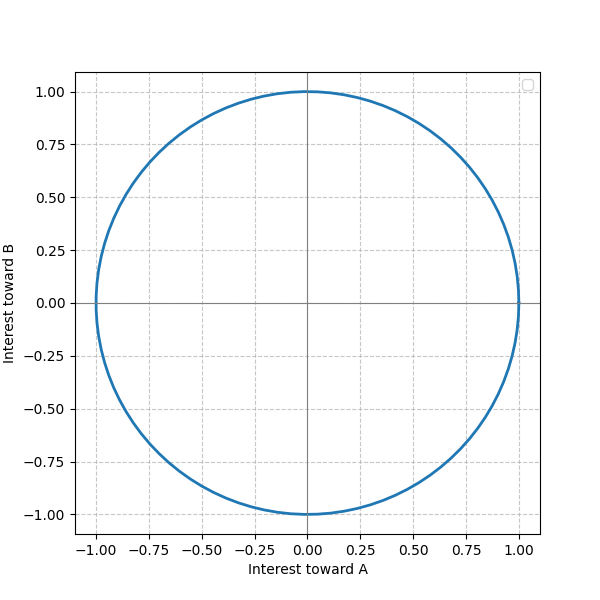}}
\subfloat[\centering]{\includegraphics[width=7.0cm]{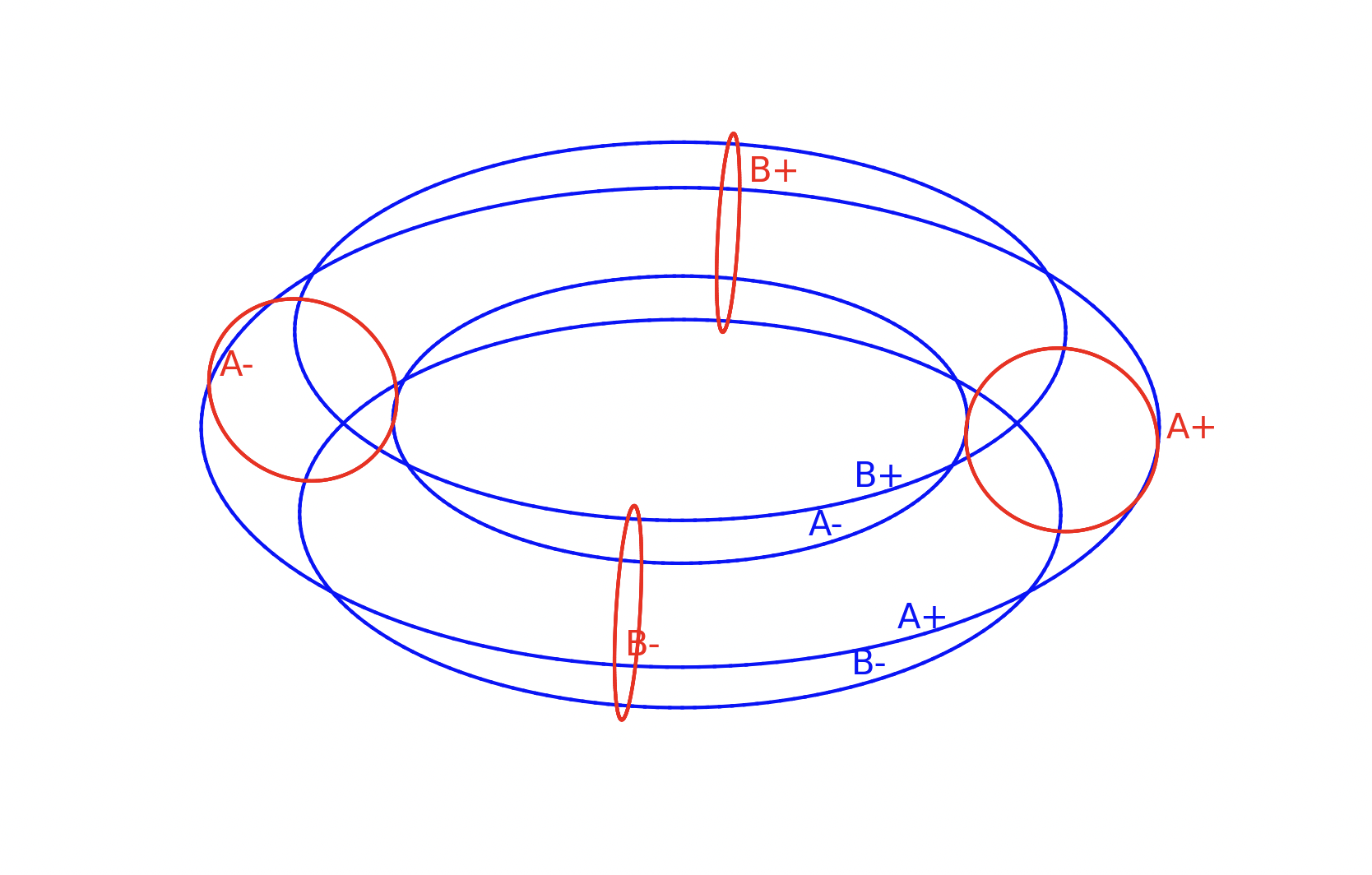}}
\caption{Two-player preference space representation. (\textbf{a}) Circle of preferences of a player. (\textbf{b}) Torus with cardinal circles. \label{fig:2pRPS}}
\end{figure} 

The preference space can be represented numerically by a two-by-two matrix like:

\begin{equation}
V = \begin{bmatrix}
v_{aa} & v_{ab} \\
v_{ba} & v_{bb} \\
\end{bmatrix} =
\begin{bmatrix}
cos(\alpha) & sin(\alpha) \\
cos(\beta) & sin(\beta) \\
\end{bmatrix}
\end{equation}

Moreover, since the rows of the matrix are unit vectors, we can replace the entries with the trigonometric functions of two angles \(\alpha\) and \(\beta\). In this way, we can reduce the dimensions and plot the preference space onto a square (this is only possible for two players).

\subsection{Application on Games}

\subsubsection{Prisoner's Dilemma}

Figure~\ref{fig:pd_outcomes} shows the
projection of the four outcomes of a Prisoner’s Dilemma onto the preference
space. Each panel corresponds to one outcome, with the red dot marking the
center of mass of the distribution and its indices $R$ (reciprocity) and $H$
(hierarchy) reported below. We used the following version of the Prisoner's Dilemma:

\[
\begin{array}{c|cc}
 & \text{L} & \text{R} \\
\hline
\text{U} & (2,2) & (0,3) \\
\text{D} & (3,0) & (1,1) \\
\end{array}
\]

\begin{figure}
\subfloat[\centering]{\includegraphics[width=7.0cm]{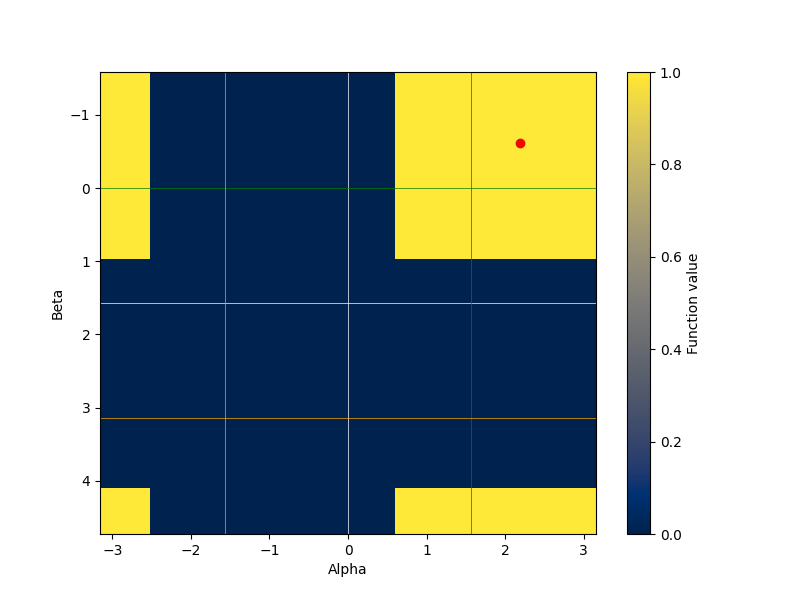}}
\subfloat[\centering]{\includegraphics[width=7.0cm]{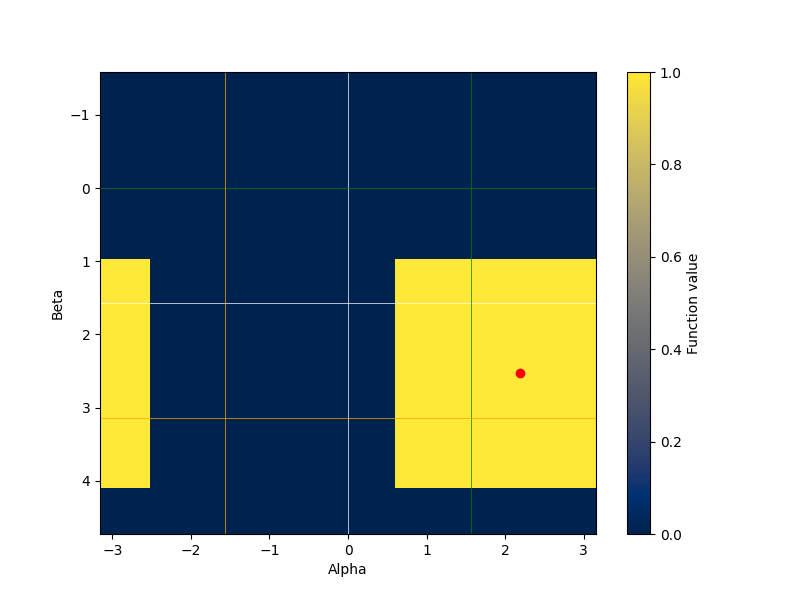}}\\
\subfloat[\centering]{\includegraphics[width=7.0cm]{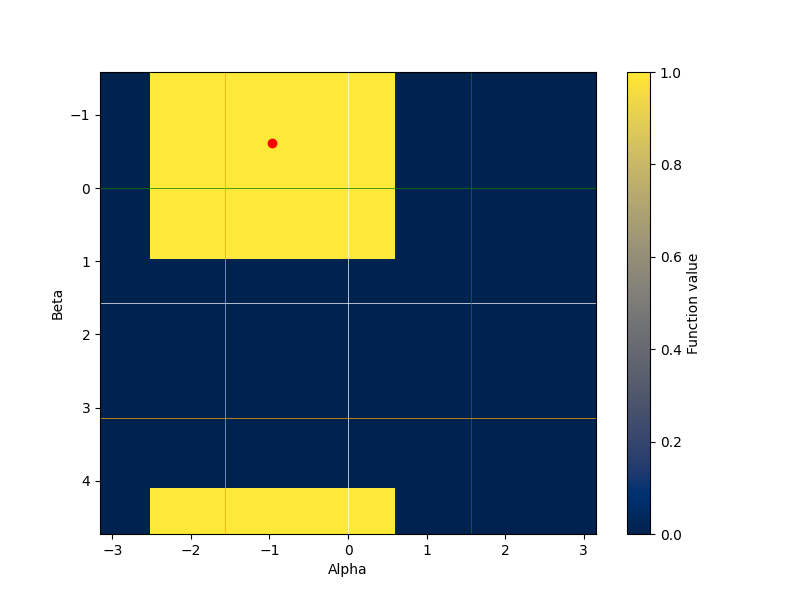}}
\subfloat[\centering]{\includegraphics[width=7.0cm]{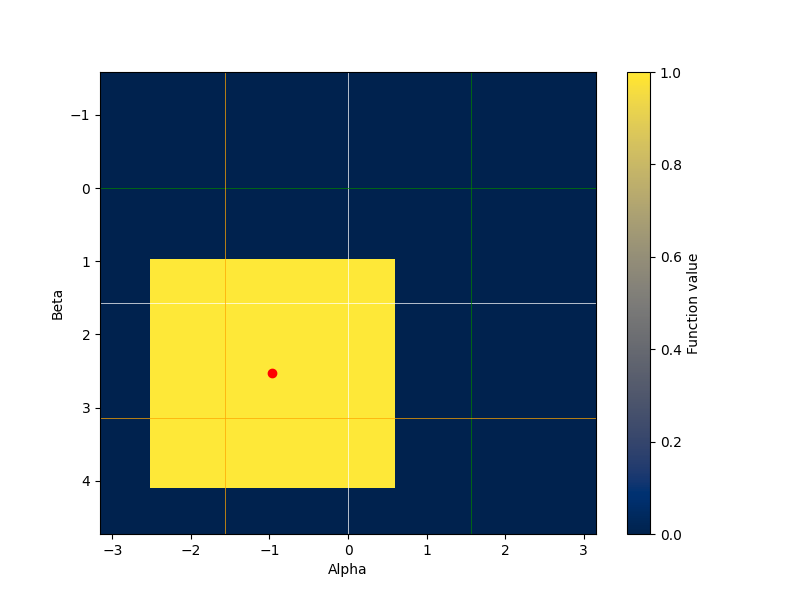}}
\caption{Projection of the four outcomes of a PD onto the preference space. Note that each outcome is positioned in a different zone. The indices are coherent with our intuitions: (UL) is positively reciprocal, (UR) and (DL) are hierarchical, (DR) is negatively reciprocal. (\textbf{a}) Projection of the outcome (UL): $R^+=0.66,\, H=0$. (\textbf{b}) Projection of the outcome (UR): $R=0,\, H=0.67$. (\textbf{c}) Projection of the outcome (DL): $R=0,\, H=0.67$. (\textbf{d}) Projection of the outcome (DR): $R^-=0.66,\, H=0$. \label{fig:pd_outcomes}}
\end{figure} 

This example highlights three important points:
\begin{enumerate}
    \item Outcomes occupy different zones of the preference space, reflecting
    their relational structure.
    \item The indices $R$ and $H$ provide a concise quantitative classification:
    reciprocity (positive or negative) versus hierarchy. 
    \item The method aligns with intuition: cooperative outcomes cluster in
    reciprocal regions, while exploitative or asymmetric outcomes appear in
    hierarchical or antagonistic zones.
\end{enumerate}

The analysis can also be conducted from the perspective of the players’
payoff functions. Figure~\ref{fig:PD_payoff} shows the projection of the two
players’ payoffs. Their centers of mass lie in different regions, shifted away
from their expected maxima, revealing asymmetries in how each player benefits
from the relational structure.
From the aggregation of the players’ CoMs we obtain the Global Center of Mass
(GCoM) of the game. For the Prisoner’s Dilemma in Figure~\ref{fig:pd_outcomes}
and Figure~\ref{fig:PD_payoff}, the GCoM yields the indices
\[
H = 0, 
\qquad R^{-} = 0.68,
\]
indicating a purely antagonistic relation with no hierarchical component. Note that, while every PD has $H=0$ and $R^->0$, the exact value of $R$ depends on the value of the payoff matrix. This not only matches the intuitive characterization of the game as one of mutual
defection and negative reciprocity, but also allows to characterize different PDs for their grade of antagonism.

\begin{figure}
\subfloat[\centering]{\includegraphics[width=7.0cm]{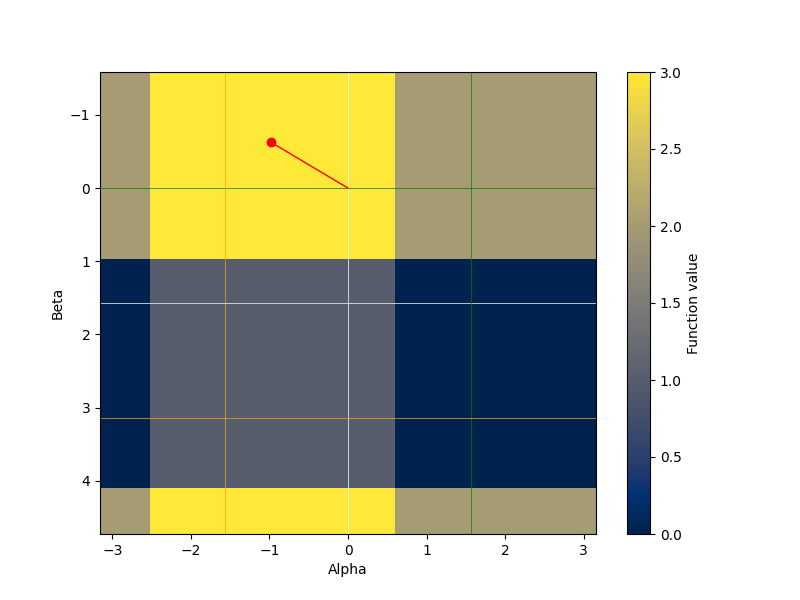}}
\subfloat[\centering]{\includegraphics[width=7.0cm]{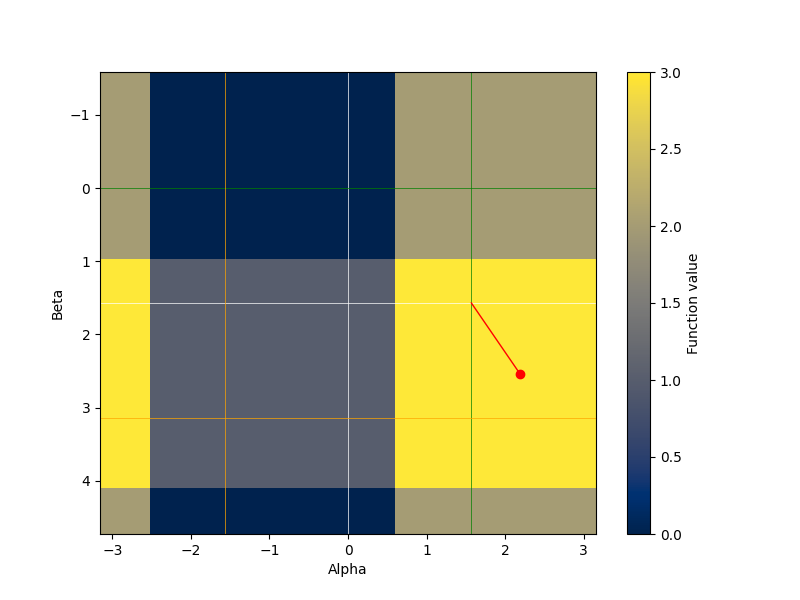}}
\caption{Projection of the payoff functions of the two players in a PD. Note that they lie in different sectors. Note also the shift of their centers of mass from their expected maximum (highlighted by a red line): for player $a$, the expected maximum is at the intersection between the white vertical and the green horizontal line $(0,0)$; for player $b$, at the intersection between the white horizontal and the vertical green line $(\frac{\pi}{2}, \frac{\pi}{2})$. (\textbf{a}) Projection of the payoff of player $a$. (\textbf{b}) Projection of the payoff of player $b$. \label{fig:PD_payoff}}
\end{figure} 

\begin{figure}
\includegraphics[width=10.0 cm]{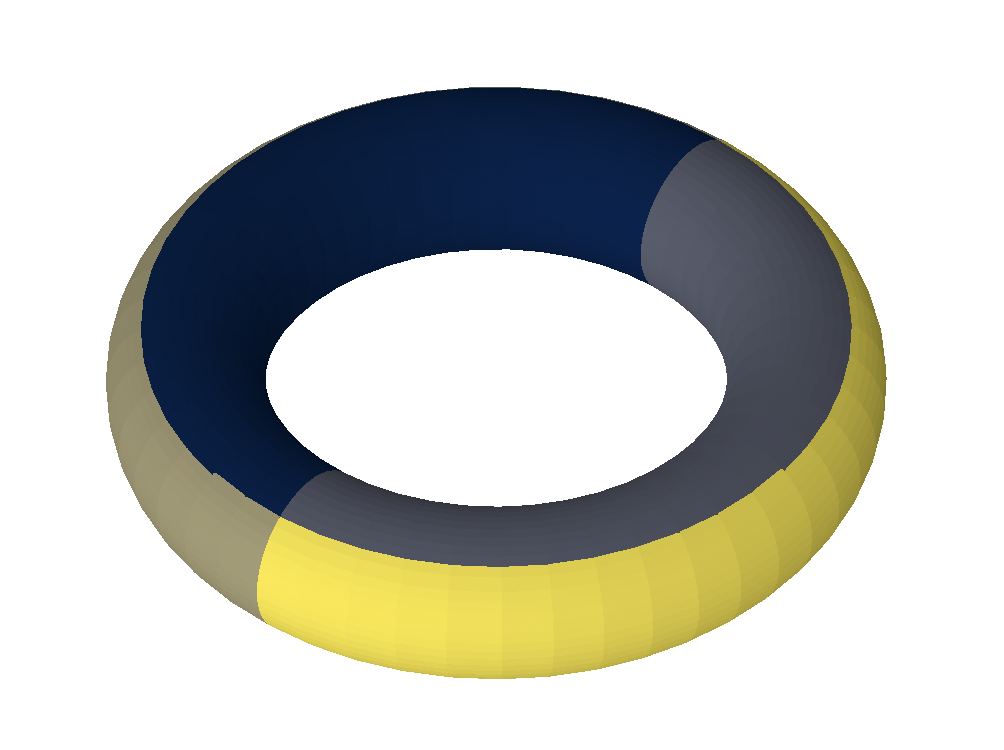}
\caption{Torus projection of the payoff of player $a$ in a PD.\label{fig:PDtorus}}
\end{figure}   
\unskip

\subsubsection{Battle of the Sexes}

The Battle of the Sexes is a classical coordination game, and thus we expect to observe positive reciprocity.  
We consider the following version of the game:
\[
\begin{array}{c|cc}
 & \text{L} & \text{R} \\
\hline
\text{U} & (3,2) & (0,0) \\
\text{D} & (0,0) & (2,3) \\
\end{array}
\]

As shown in Figure~\ref{fig:BS_payoff}, the two centers of mass lie in the same region of the preference space, where the players’ preferences are positively correlated.  
The GCoM in this case yields an index \(R^+ = 0.38\), while \(H = 0\), since the game is symmetric.  
The reciprocity index \(R\) can naturally be increased by adding a constant value to the payoffs of both players in outcomes \((U,L)\) and \((D,R)\).  

In Figure~\ref{fig:BS_outcome} we observe that the region of the preference space where a single outcome constitutes an equilibrium is located precisely between two cardinal points (for instance, the yellow tongue between the white and orange vertical lines).  
This demonstrates that a discrete preference space, such as the one considered in~\citep{GrossiTurrini2012}, is insufficient to capture the full range of possible mixed coalitions between players.  

\begin{figure} 
\subfloat[\centering]{\includegraphics[width=7.0cm]{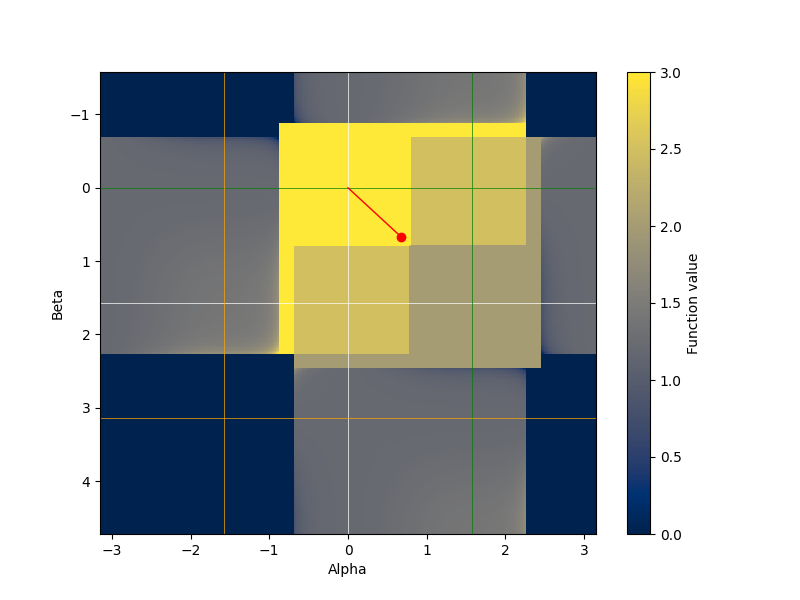}}
\subfloat[\centering]{\includegraphics[width=7.0cm]{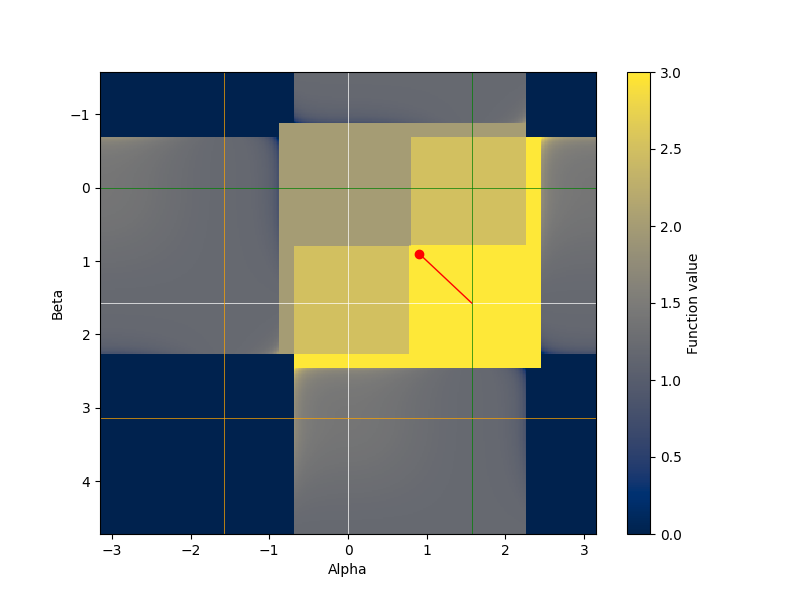}}
\caption{Projection of the payoff functions of the two players in a Battle of the Sexes. Note that they lie in the same sector. (\textbf{a}) Projection of the payoff of player $a$. (\textbf{b}) Projection of the payoff of player $b$. \label{fig:BS_payoff}}
\end{figure} 

\begin{figure}
\includegraphics[width=10.0 cm]{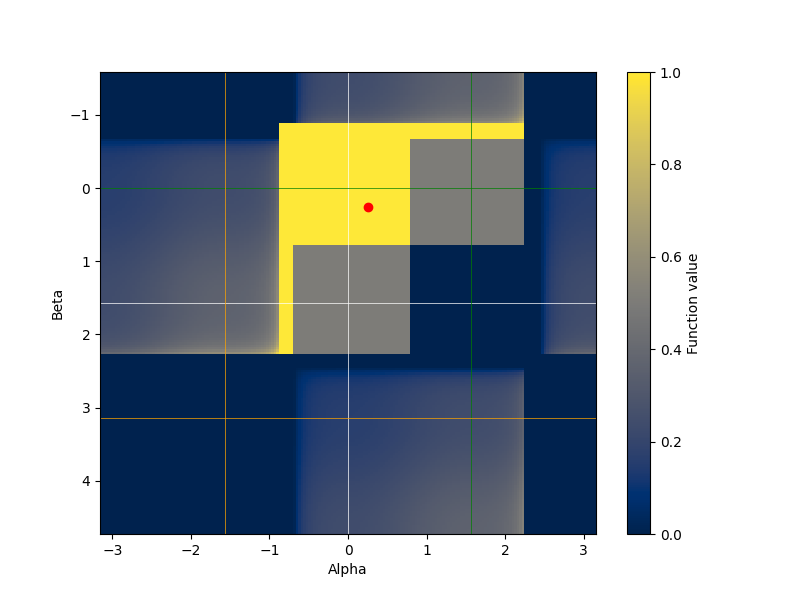}
\caption{Battle of the Sexes. Projection of the outcome (U,L) in the preference space. Note that some significant parts of the yellow region lie in between two cardinal points.\label{fig:BS_outcome}}
\end{figure}   
\unskip

\subsubsection{Asymmetric Games}

Only asymmetric games exhibit a non-zero value of the \(H\) index.  
The graphs corresponding to the two players are therefore necessarily distinct.  
Figure~\ref{fig:Asymmetric} presents a simple asymmetric game with the following payoff matrix:

\[
\begin{array}{c|cc}
 & \text{L} & \text{R} \\
\hline
\text{U} & (1,1) & (3,0) \\
\text{D} & (0,2) & (2,3) \\
\end{array}
\]

This game displays an intriguing property.  
As usual, the standard equilibrium is given by the intersection of the two white lines.  
However, in this case we observe a paradox: the expected payoff of player~\(a\) at this point is strictly lower than the payoff attained at the two extreme points of the same horizontal line (i.e., without any change in player~\(b\)’s preference position). In other words, we have an index of diagonal power below zero, namely $\rho_{aa}=-\frac{1}{3}$. The reader may verify that player~\(a\) achieves a better outcome by \emph{minimizing}, rather than maximizing, their own payoff.  
We refer to this phenomenon as the \emph{Self-Harm Paradox}.   

\begin{figure}
\subfloat[\centering]{\includegraphics[width=7.0cm]{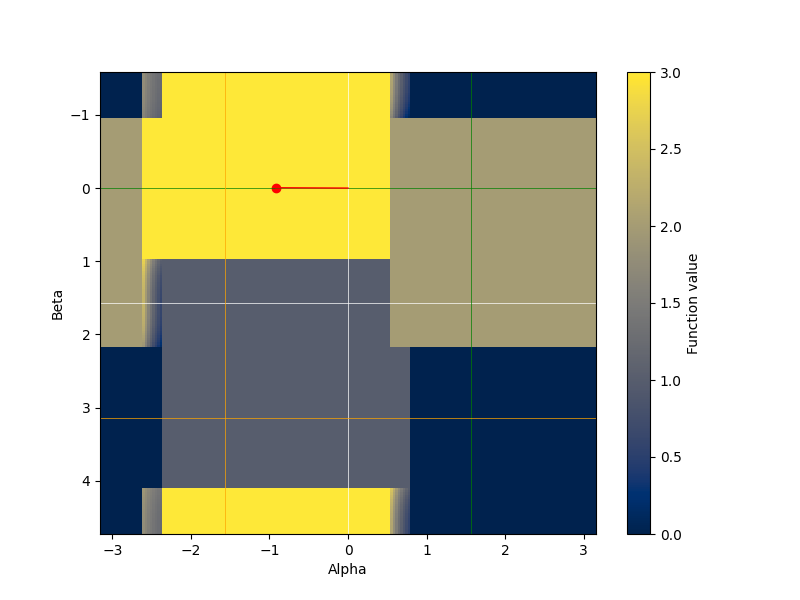}}
\subfloat[\centering]{\includegraphics[width=7.0cm]{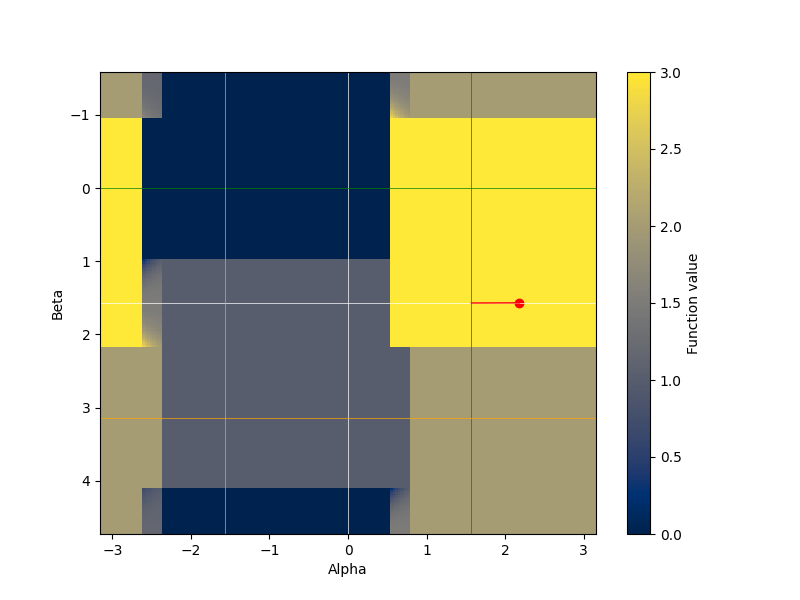}}\\
\caption{Projection of the two players' payoffs in a \textit{Self-Harm Paradox}. (\textbf{a}) Projection of the payoff of player $a$. (\textbf{b}) Projection of the payoff of player $b$. \label{fig:Asymmetric}}
\end{figure} 

\subsubsection{Matching Pennies and Rock--Paper--Scissors}

The preference space and the CoM provide a classification that goes beyond the traditional typologies of games.  
They enable a deeper analysis that not only refines distinctions among various forms of the Prisoner’s Dilemma or the Battle of the Sexes, but also permits the comparison of seemingly different games within a uniform system of measurement.  
This feature becomes particularly evident when the number or type of strategies (discrete or continuous) is varied.  

\[
\begin{array}{c|cc}
    \ (1,0) & (0,1) \\
    \ (0,1) & (1,0) \\
\end{array}
\hspace{3cm}
\begin{array}{c|c|cc}
    \ (1,1) & (0,2) & (2,0) \\
    \ (2,0) & (1,1) & (0,2) \\
    \ (0,2) & (2,0) & (1,1) \\
\end{array}
\]
\[
\text{Matching Pennies} \hspace{3cm} \text{Rock--Paper--Scissors}
\]

Figure~\ref{fig:MPRPS} illustrates the classic Matching Pennies and Rock--Paper--Scissors games.  
Although their strategy spaces differ---the former offering two strategies per player and the latter three---the projection of payoffs onto the preference space and the resulting positions of the CoM are identical, since the two games are structurally equivalent.  
This provides further evidence of the independence of the preference space from the strategic level.  

\begin{figure}
\subfloat[\centering]{\includegraphics[width=7.0cm]{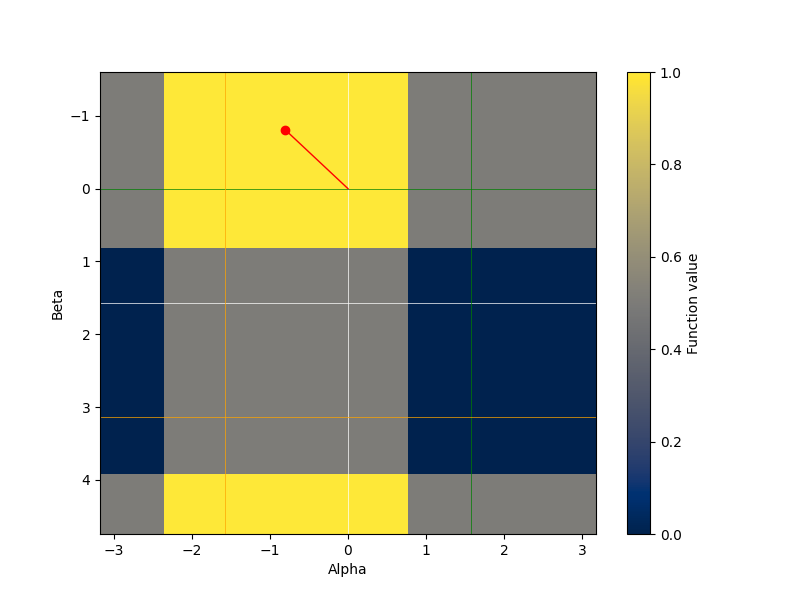}}
\subfloat[\centering]{\includegraphics[width=7.0cm]{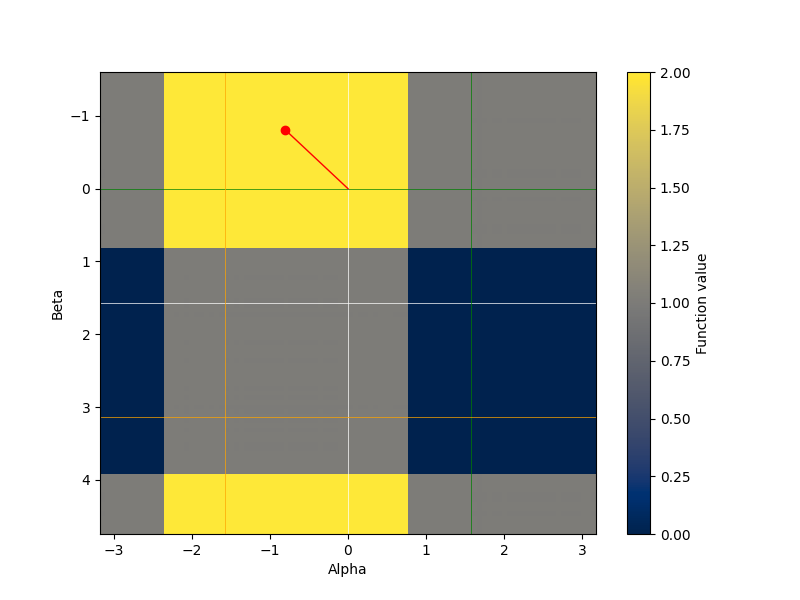}}\\
\caption{Projection of the payoff of player $a$ in: (\textbf{a}) Matching Pennies, (\textbf{a}) Rock--Paper--Scissors. \label{fig:MPRPS}}
\end{figure} 

\subsubsection{Economic model: Cournot}

In Figure~\ref{fig:Cournot} the reader can appreciate the fact that preference space is able to represent in the same space discrete games and continuous economic models. Here a classic Cournot Duopoly is analyzed. The price function is \( P=(a-Q)\) with \(Q=q_1+q_2\) and the profit functions are \(\pi_i=q_i(P-c)\) where \(c\) is the cost. The values of \(a\) and \(c\) are set respectively to \(10\) and \(2\). We also imposed an upper and lower limit for the quantities: \(0<q_i<(a-c)/2\).

The classical Cournot equilibrium is placed at the intersection of the two white lines.
The CoM is placed in the left region, and we will have the same situation for the other firm. The position of the CoM shows the lightly competitive nature of the model. The GCoM gives $R^-=0.06$ and $H=0$.

\begin{figure}
\subfloat[\centering]{\includegraphics[width=7.0cm]{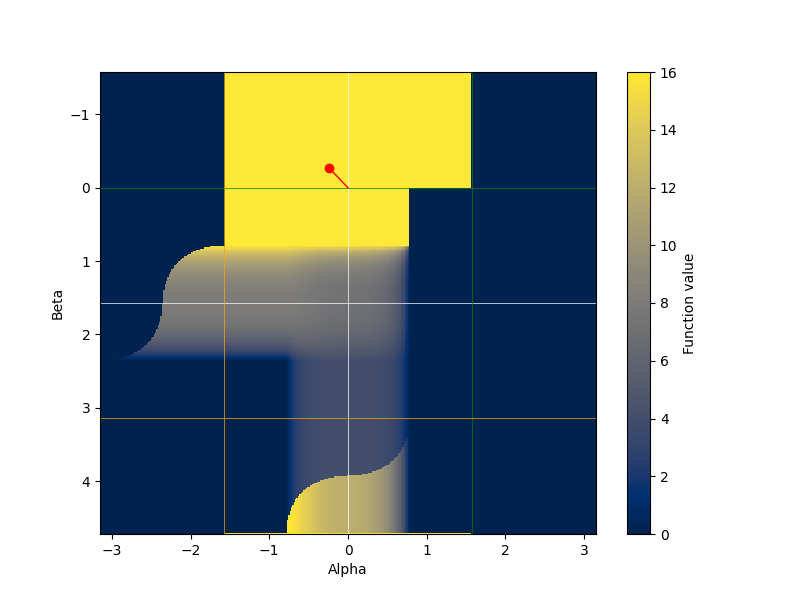}}
\subfloat[\centering]{\includegraphics[width=7.0cm]{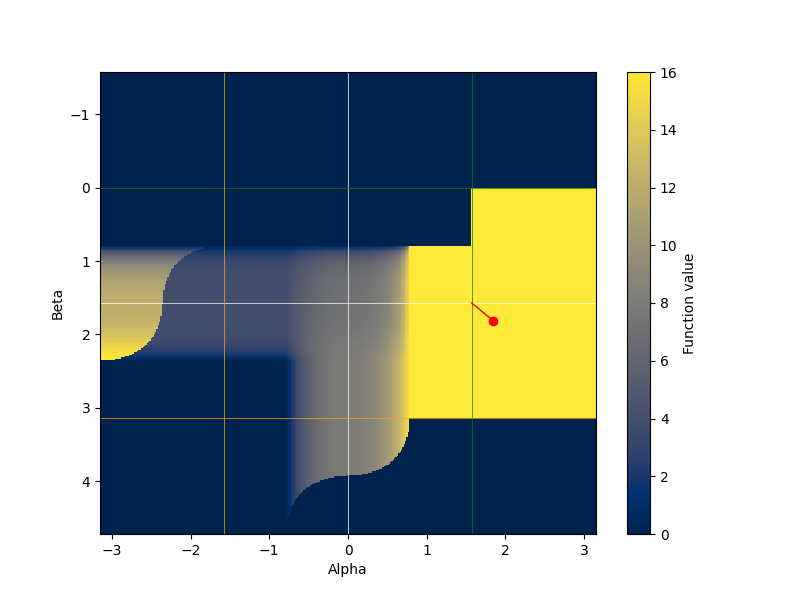}}\\
\caption{Projection of the payoff functions of the two players in a Cournot Duopoly. (\textbf{a}) Projection of the payoff of player $a$. (\textbf{b}) Projection of the payoff of player $b$. \label{fig:Cournot}}
\end{figure} 

\subsubsection{3-Players Games}

\begin{figure}
\subfloat[\centering]{\includegraphics[width=7.0cm]{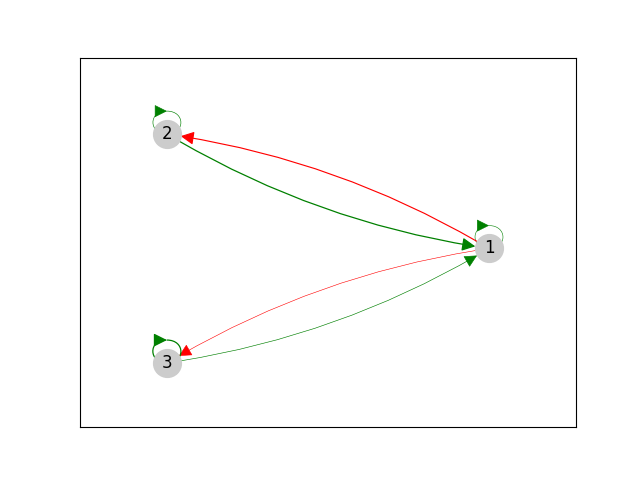}}
\subfloat[\centering]{\includegraphics[width=7.0cm]{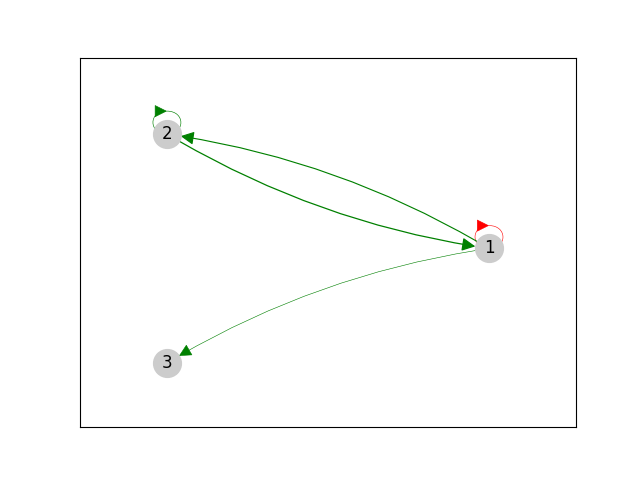}}\\
\subfloat[\centering]{\includegraphics[width=7.0cm]{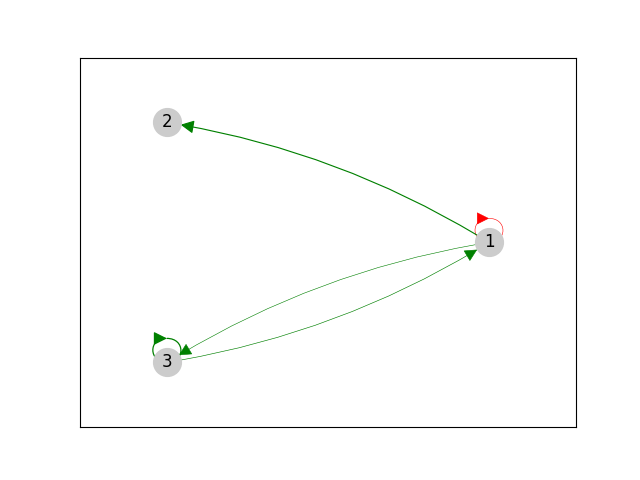}}
\subfloat[\centering]{\includegraphics[width=7.0cm]{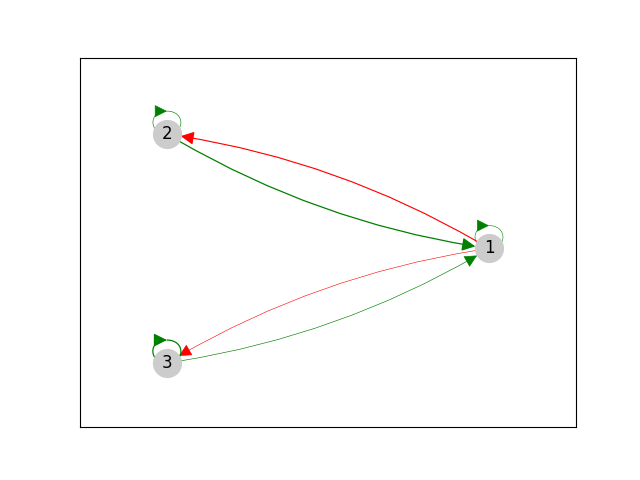}}
\caption{The CoMs of the payoff of the players and the GCoM of a 3-players game. (\textbf{a}) The CoM of player $1$. (\textbf{b}) The CoM of player $2$. (\textbf{c}) The CoM of player $3$. (\textbf{d}) The GCoM. \label{fig:A3}}
\end{figure} 

We cannot represent the entire preference space of games with more than two players as in the two-player case.  
However, we can represent a single point in that space---in particular, any CoM---by means of an adjacency graph.  
The CoM of a player’s payoff conveys information not only about their direct relations with other players, but also about indirect ones.  
For instance, it allows us to address questions such as: what relation between players \(j\) and \(k\) is preferred by player \(i\)?  
Does player \(i\) favor antagonism between them, or rather hierarchy, or mutualism?  

As an illustrative case, we consider the following three-player game:

\[
\begin{array}{c|cc}
 & \text{L} & \text{R} \\
\hline
\text{U} & (3,1,0) & (1,0,0) \\
\text{D} & (2,3,1) & (0,2,1) \\
\end{array}
\hspace{2cm}
\begin{array}{c|cc}
 & \text{L} & \text{R} \\
\hline
\text{U} & (4,1,2) & (2,0,2) \\
\text{D} & (3,3,3) & (1,2,3) \\
\end{array}
\]
\[
\text{(Player 3 chooses A)} \hspace{3cm} \text{(Player 3 chooses B)}
\]

We choose this game because it has a straightforward and clear structure. Player $2$ and $3$ have no influence over each other, so $\rho_{23}=\rho_{32}=0$. There exists a unique equilibrium in dominant strategies, namely \((U,L,B)\). Moreover, and this is important for the preference space, $(U)$ is a dominated strategy for player $2$ and $3$, $(R)$ and $(B)$ are dominated strategy for player $1$ (remember that permutation point are significant).

Figure~\ref{fig:A3} displays the CoM of each player together with the GCoM.  
Remember that the column corresponding to the focal player in each CoM always coincides with the expected maximum point of that player and therefore carries limited information, being invariably positive.  
In the figures, negative values are represented by red arrows, positive values by green arrows, and the width of each arrow is proportional to its absolute weight.  

It is worth noting that the graph of player~2 and ~3 (Figure~\ref{fig:A3}b and Figure~\ref{fig:A3}c) contain an edge between respectively players~1 and~3 and players~1 and~2, i.e. indirect preferences.  
This implies that player $2$ could, in general, increase their payoff if player $1$ favors $3$, and the same for $3$ if $1$ favors $2$. There is an indirect sympathy between player $2$ and $3$, despite they are strategically unrelated, and this happens because player $1$'s strategy $(D)$ is favorable to both.

The self-loops in the CoM graphs represent the auto-orientation of the other two players in directions favorable to the focal one.  
For instance, in Figure~\ref{fig:A3}a both players~2 and~3 exhibit positive self-loops, whereas it is favorable for player~1 to reduce the payoffs of players~2 and~3.  
This configuration clearly illustrates a hierarchical relation with player~1 occupying the dominant position.  
Consistently, the inverse relations (Figures~\ref{fig:A3}b and~\ref{fig:A3}c) display opposite signs.  

Finally, the GCoM highlights the hierarchical structure of the game, with indices \(H=0.61\) and \(R=0\). 
It should be observed that the GCoM and the CoM of player $1$ are practically identical: in other words, the dominium of player $1$ is depicted in the coincidence between their ideal structure and the collective one.


\section{Discussion}

We have proposed a geometric framework for analyzing power relations in games by embedding players' social preferences into a continuous preference space. This construction moves beyond traditional classifications and provides a structural description of relational attitudes that is independent of the specific strategic form. Our results demonstrate that the distribution of every function defined in the preference space reveals meaningful properties, and that these properties can be effectively summarized through the center of mass and our proposed structural indices.

\subsection{Theoretical Contributions}

The preference space framework makes several key contributions to the analysis of power in games. First, it provides a \textit{canonical domain} for representing relational attitudes, eliminating the arbitrariness in selecting utility functions. By constructing a projective space, we obtain a natural geometry that captures the full spectrum of possible relational stances—from pure altruism to pure antagonism—without ad hoc discretization.

Second, our framework \textit{unifies cooperative and non-cooperative approaches} to power analysis. The preference space naturally links to coalitional thinking through Equation~\ref{coalitional}, permutation points (see Section~\ref{sec:permutation}) and bargaining power (see Theorem~\ref{th:voting}), allowing us to recover classical cooperative solution concepts while maintaining the rich strategic structure of non-cooperative games. This bridges what have traditionally been separate literatures, showing how power emerges from the interplay between strategic possibilities and relational attitudes.

Third, the framework reveals \textit{structural similarities across seemingly different games}. The identical projection of Matching Pennies and Rock-Paper-Scissors (Figure 8) demonstrates that games with different strategy sets can share the same underlying power structure. This suggests that our method captures fundamental relational properties that are independent of the specific strategic implementation.

\subsection{Limitations and Boundary Conditions}

Some limitations of our framework warrant discussion. 

The computational complexity of exploring the preference space grows exponentially with the number of players. While the two-player case can be visualized on a torus, higher-dimensional spaces require Monte Carlo sampling or other approximation techniques, as we employed in Section 3.3.6.

The framework relies on the properties of the $\mu$ function (Section 2.1), particularly the Pareto Efficiency axiom used in Theorem~\ref{th:expectedmax} and Corollary~\ref{co:expectedmin}. While this axiom is reasonable for analyzing cooperative potential, alternative specifications of $\mu$—focusing on risk-dominant equilibria or using different selection criteria—might yield different landscapes in the preference space. Future work could explore how sensitive our results are to variations in $\mu$.

\subsection{Future Research Directions}
\label{subsec:future_research}

\subsubsection{Dynamic Games and Evolving Preferences}
A primary extension involves moving from static to \textit{dynamic games}. The current framework analyzes power at a single point in time, but many strategic interactions unfold over multiple periods. Future work could model how the preference space $\mathbb{V}$ itself evolves based on past outcomes, learning, or reputation. For instance, a player's tendency to reciprocate cooperation or punish antagonism could be formalized as a dynamic path within $\mathbb{V}$. The Center of Mass (CoM) could then track the evolution of power relations, revealing whether interactions converge to stable hierarchical or reciprocal structures over time. This would address a key limitation of static axiomatic measures, which cannot capture the endogenous evolution of power through play.

\subsubsection{Incomplete Information and Beliefs}
Our framework currently assumes complete information. A critical next step is to incorporate \textit{incomplete information} about players' relational preferences. This would involve defining a Bayesian version of the preference space, where players have beliefs about others' $v_i$ vectors. Power would then stem not only from actual preferences but also from perceived ones. Questions of signaling, trust, and the manipulation of beliefs could be analyzed geometrically. For example, how does a player's ability to signal altruism affect their bargaining power? This extension would bridge our geometric approach with the literature on Bayesian persuasion and information design.

\subsubsection{Empirical Estimation and Behavioral Applications}
A major advantage of our geometric framework is its potential for \textit{empirical application}. The indices $D$, $H$, and $R$ provide testable hypotheses about the relational structure of observed interactions. Future work could develop methods to estimate a player's position in the preference space from behavioral data, such as experimental game play or field data on negotiations. 

\subsubsection{Institutional Design and Comparative Analysis}
Our framework can be powerfully applied to questions of \textit{institutional design}. By projecting different institutional rules (e.g., varying amendment procedures in legislatures or veto powers in committees) onto the preference space, one can compare how they shape the resulting distribution of power, as summarized by the GCoM and its indices. This would provide a geometric tool for assessing whether a proposed rule change promotes greater reciprocity or reinforces hierarchy. Such an analysis moves beyond the ex-post measurement of power to an ex-ante evaluation of how rules structure relationships, a key concern for political scientists and mechanism designers.

\subsubsection{Algorithmic and Computational Advancements}
As the number of players grows, the dimension of the preference space $\mathbb{V}$ increases, posing computational challenges. Future work could develop more efficient algorithms for sampling high-dimensional preference spaces and computing the associated centers of mass. Techniques from machine learning, such as variational inference, could be employed to approximate the CoM without exhaustive sampling. This would make the framework applicable to large-scale strategic networks, such as those studied in organizational sociology or international relations.

\newpage

\bibliographystyle{apalike}
\bibliography{RPS_power_bibliography}

\end{document}